\documentclass{article}
\usepackage{spconf,amsmath,graphicx,hyperref}

\usepackage{comment}
\usepackage[ruled,vlined]{algorithm2e}
\usepackage{graphicx}
\usepackage{subcaption}
\usepackage{float}
\usepackage{booktabs}
\usepackage{xcolor}
\usepackage{multirow}
\usepackage{amsmath,amssymb,amsthm,amsfonts}
\usepackage{mathtools}
\usepackage{xcolor}             %
\usepackage{amssymb}
\usepackage{amsfonts}
\usepackage[noend]{algpseudocode}
\usepackage{hyperref}
\usepackage{bbm}
\usepackage{amsthm}

\newtheorem*{proposition*}{Proposition}

\title{Principled Coarse-Grained Acceptance \\  for Speculative Decoding in Speech}
\name{}
\address{}
\name{Moran Yanuka$^{1,2}$, Paul Dixon$^{1}$, Eyal Finkelshtein$^{1}$, Daniel Rotman$^{1}$, Raja Giryes$^{2}$} 
\address{%
$^1$ Apple, 
$^2$ Tel-Aviv University%
}

\begin{document}

\maketitle
\newcommand{\targetdist}{p}
\newcommand{\draftdist}{q}
\newcommand{\targetprob}[1]{\targetdist(#1)}
\newcommand{\draftprob}[1]{\draftdist(#1)}
\newcommand{\targetgroupprob}[1]{P(#1)}
\newcommand{\draftgroupprob}[1]{Q(#1)}
\newcommand{\embfunc}[2]{\mathrm{Emb}_{#1}(#2)}
\newcommand{\seg}[1]{\mathcal{G}(#1)}
\newcommand{\vocab}{\mathcal{V}}

\newcommand{\paul}[1]{\textcolor{purple}{[Paul: #1]}}

\newcommand{\moran}[1]{\textcolor{orange}{[Moran: #1]}}

\newcommand{\todo}[1]{\textcolor{red}{[TODO: #1]}}

\begin{abstract}
Speculative decoding accelerates autoregressive speech generation by letting a fast draft model propose tokens that a larger target model verifies. However, for speech LLMs that generate acoustic tokens, exact token matching is overly restrictive: many discrete tokens are acoustically or semantically interchangeable, reducing acceptance rates and limiting speedups. We introduce Principled Coarse-Graining (PCG), which verifies proposals at the level of Acoustic Similarity Groups (ASGs) derived from the target model’s embedding space. By splitting each token’s probability mass across the overlapping groups that contain it, we define an overlap-aware coarse-grained distribution and perform rejection sampling on the resulting group variable. This yields an exactness guarantee at the group level while allowing the accepted draft token to stand in for any member of the group in practice. On LibriTTS, PCG increases acceptance and throughput relative to standard speculative decoding and prior speech-specific relaxations while maintaining intelligibility and speaker similarity. These results suggest acoustically aware, group-level acceptance as a simple and general way to accelerate speech token generation while maintaining speech quality.
\end{abstract}


\section{Introduction}

Speculative decoding (SD) is a two-model approach for accelerating autoregressive decoding in large language models \cite{chen2023acceleratinglargelanguagemodel, leviathan2023fastinferencetransformersspeculative}. A small, fast draft model proposes token sequences that a larger, more accurate target model subsequently verifies or corrects. SD can deliver large wall-clock speedups while provably preserving the target model’s output distribution. This property is especially valuable for on-device TTS, where strict latency budgets and offline operation for privacy constrain compute.

\begin{figure}[t]
    \centering
    \includegraphics[width=1\columnwidth]{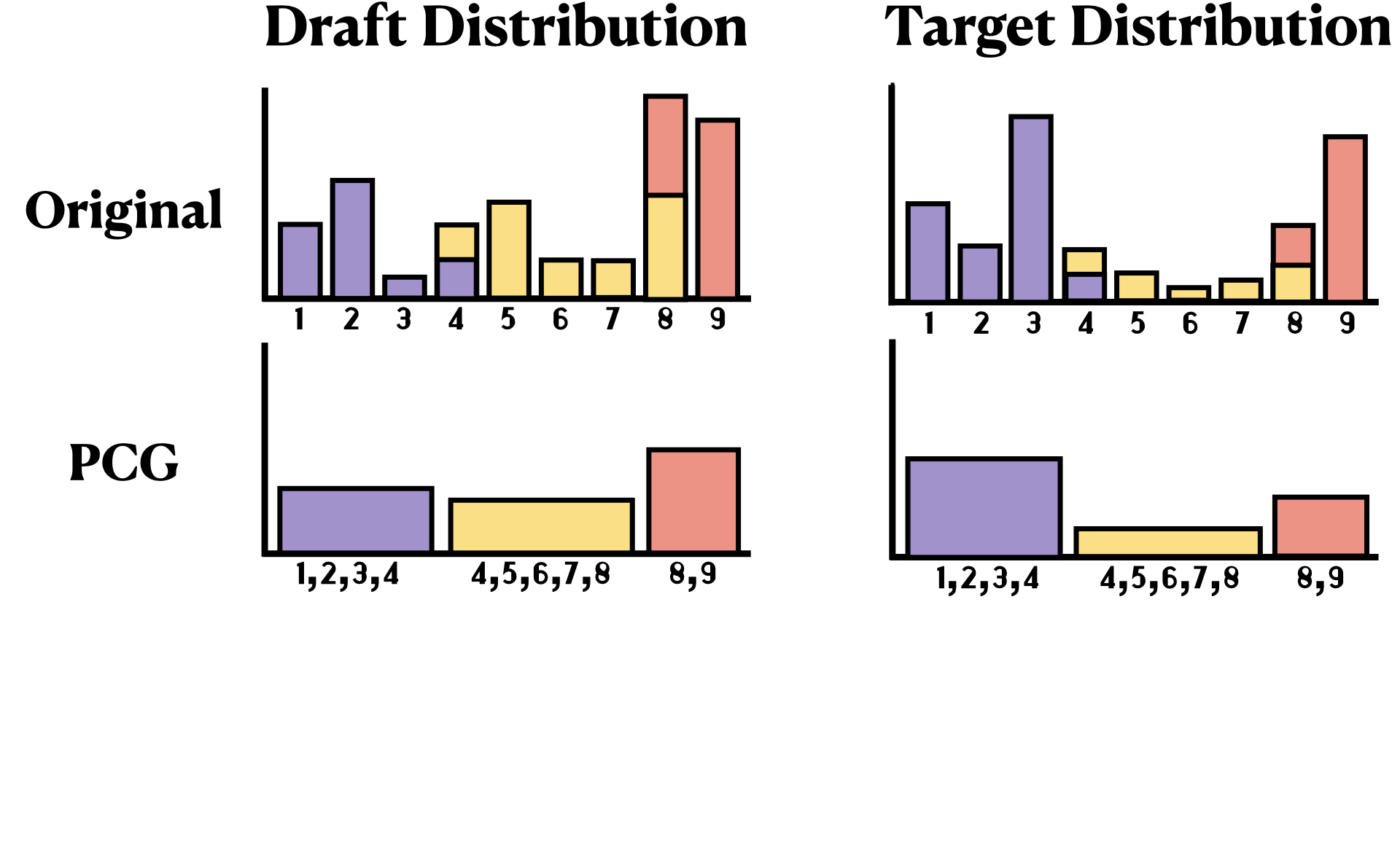}
    \vspace{-5.5em}
    \caption{Token-level probabilities used for the SD acceptance criterion: original vs PCG. Colors denote Acoustic Similarity Groups (ASGs); tokens may belong to multiple groups.}
    \label{fig:method}

\end{figure}

Applying SD to autoregressive models over speech tokens poses distinct challenges. Many speech tokens are acoustically or phonetically interchangeable, so different tokens can yield perceptually similar audio~\cite{lin25h_interspeech,liu-etal-2025-analyzing,10888194}. Because standard SD accepts only exact token matches, it fails to exploit this many-to-one mapping and rejects draft tokens that are perceptually valid. The resulting low acceptance rates can erase the speed gains from accepted tokens.

Prior adaptations of SD to speech synthesis address low acceptance by restricting sampling. Methods such as sequence-level Viterbi search~\cite{nguyen2025accelerating} and tolerance-based validation~\cite{10888194} typically confine decoding to small top-k pools, which reduces output diversity~\cite{zhou-etal-2025-balancing}. SSD~\cite{lin25h_interspeech} instead increases acceptance by adding a constant probability bias, enabling long-tail sampling but ignoring acoustic similarity and risking erroneous acceptances. For visual AR models, LANTERN~\cite{Jang2025LANTERN} relaxes acceptance by aggregating probability over latent neighbors and bounding distortion via total variation distance. In contrast, we introduce an explicit coarse-grained random variable and run acceptance/rejection on its induced distribution, guaranteeing exact sampling at the coarse level. While coarse-graining is well studied in statistical physics and explored in ML~\cite{schweitzer2021survey,stephan2025coarse,tian2024visual}, we are not aware of SD methods that define a group variable with overlap-aware acceptance and perform rejection sampling on its valid induced distribution.

\begin{figure*}[t]
    \centering
    \includegraphics[width=\textwidth]{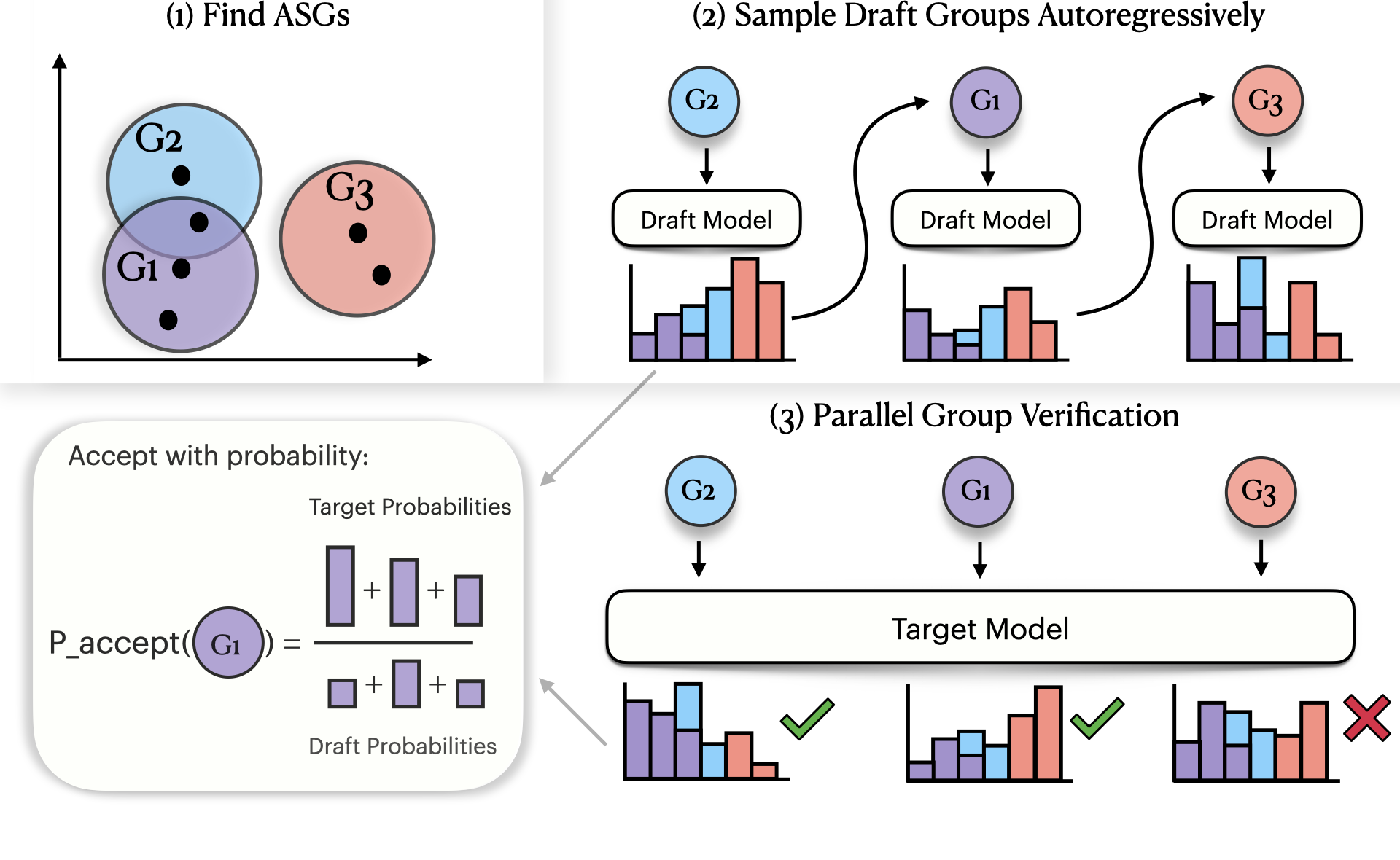}
    \caption{\textbf{Overview of speculative decoding with PCG.} Tokens are first clustered into overlapping ASGs in the target model’s embedding space. The lightweight draft model then autoregressively samples these groups. Finally, the target model verifies all proposed groups in parallel.}
    \label{fig:full_method}

\end{figure*}

We propose \emph{Principled Coarse-Graining} (PCG), a framework that replaces exact token matching with group-level verification. We construct \emph{Acoustic Similarity Groups} (ASGs) in the target model’s token embedding space, capturing its internal organization of semantic and acoustic similarity. PCG performs speculative sampling on the coarse-grained distribution over ASGs and carries out rejection sampling at the group level. This yields a strong guarantee: the sequence of accepted groups is an exact sample from the target model’s coarse-grained distribution at each decoding step. Empirically, draft predictions are often correct at the group level, which substantially raises acceptance while maintaining diversity.

Figure \ref{fig:method} illustrates the intuition. The top row shows draft and target distributions over individual tokens; the bottom row shows the same distributions after coarse-graining into ASGs. Under strict matching a draft sample like token 2 is likely rejected, whereas under PCG its group {1,2,3,4} is accepted because the target assigns high mass to the group even if not to the specific token. Overlap (e.g., token 4 belonging to multiple groups) preserves smooth acoustic neighborhoods.

Our contributions are twofold: we introduce PCG, a group-level speculative decoding scheme that performs rejection sampling on a valid coarse-grained distribution over ASGs, guaranteeing exact coarse-level sampling at each step; and we provide a drop-in AR TTS implementation that raises acceptance and decoding speed while preserving target diversity, with ablations over group granularity and overlap.

\section{Method}
\label{sec:method}

Our framework addresses the challenges of speech token decoding through two key innovations. First, we define ASGs that capture how multiple discrete tokens map to perceptually similar sounds, with tokens belonging to multiple overlapping groups. Second, we introduce PCG, which performs group-level acceptance while accounting for these overlaps—preserving acoustic neighborhoods without arbitrary boundaries. See Figure \ref{fig:full_method} for an overview of speculative decoding with PCG. We denote the target model distribution by $q$ and the draft model distribution by $p$.

\subsection{Acoustic Similarity Groups (ASG)}
\label{sec:method_segs}

To capture the many-to-one mapping between discrete speech tokens and perceived sounds, we group tokens with similar embeddings from the target model $q$. Learned embeddings across text, vision, and audio exhibit rich semantic structure and cross-modal alignment, motivating embedding similarity as a perceptual proxy~\cite{dar-etal-2023-analyzing, 10.1109/TASLP.2023.3288409, 10.5555/3737916.3739168}. For a token $t$, its ASG is
\begin{equation}
\mathcal{G}(t) := \{ t' \in \mathcal{V} : \cos(\text{Emb}_q(t), \text{Emb}_q(t')) > \theta \}.
\label{eq:seg_construction}
\end{equation}
where $\text{Emb}_q(\cdot)$ is the target model's token embedding and $\mathcal{V}$ is the vocabulary of the model. The threshold $\theta$ controls the granularity: lower values create larger, more permissive groups. We find $\theta \in [0.38, 0.45]$ optimal in our experiments (Figure~\ref{fig:tradeoff_curve}).

The collection of distinct ASGs forms our group set $\{\mathcal{G}_k\}_{k=1}^M$, where each $\mathcal{G}_k$ corresponds to a unique equivalence group and $M \leq |\mathcal{V}|$.

\subsection{Principled Coarse-Graining (PCG)}
\label{sec:method_pcg}

ASGs naturally overlap: tokens belong to multiple groups thus preserving acoustic neighborhoods without arbitrary boundaries. However, overlap prevents naive probability aggregation. PCG resolves this by splitting each token's probability across its groups, yielding proper distributions $P_c$ and $Q_c$ for group-level acceptance.

\smallskip\noindent
\textbf{Overlap-aware coarse distributions.}
For a token $t$, let $S(t) := \{k : t \in \mathcal{G}_k\}$ and $N(t) := |S(t)|$. We choose nonnegative weights $w_{k,t}$ for $k \in S(t)$ such that
\begin{equation}
\sum_{k \in S(t)} w_{k,t} = 1, \qquad w_{k,t} = 0 \;\text{if}\; t \notin \mathcal{G}_k.
\label{eq:weights}
\end{equation}
Our default is equal-split,
\begin{equation}
w_{k,t} = \frac{\mathbbm{1}\{ t \in \mathcal{G}_k \}}{N(t)}.
\end{equation}
Given $w_{k,t}$, define the coarse-grained group distributions
\begin{equation}
P_c(\mathcal{G}_k) = \sum_{t \in \mathcal{G}_k} p(t) \cdot w_{k,t},
\qquad
Q_c(\mathcal{G}_k) = \sum_{t \in \mathcal{G}_k} q(t) \cdot w_{k,t}.
\label{eq:PcQc}
\end{equation}
Because the weights in~\eqref{eq:weights} sum to one per token, we have $\sum_k P_c(\mathcal{G}_k) = \sum_t p(t) = 1$ and $\sum_k Q_c(\mathcal{G}_k) = 1$, hence no global normalization constants are needed. Under equal-split,
\begin{equation}
P_c(\mathcal{G}_k) = \sum_{t \in \mathcal{G}_k} \frac{p(t)}{N(t)},
\qquad
Q_c(\mathcal{G}_k) = \sum_{t \in \mathcal{G}_k} \frac{q(t)}{N(t)}.
\end{equation}

\smallskip\noindent
\textbf{Proposal coupling and acceptance.}
At speculative step $i$, draw a draft token $x_i \sim p_i(\cdot)$. Because groups overlap, we then draw a group label $K_i$ from the groups containing $x_i$ using the same per-token weights:
\begin{equation}
K_i \sim w_{k,x_i} \quad\text{over}\quad k \in S(x_i).
\label{eq:group_sampling}
\end{equation}
This coupling ensures $\Pr(K_i = k) = \sum_t p_i(t) \cdot w_{k,t} = P_c^{(i)}(\mathcal{G}_k)$, i.e., the proposal over groups equals the coarse-grained draft distribution at step $i$.

The group-level acceptance ratio is then the standard speculative ratio for the same random variable (the group label):
\begin{equation}
r_i = \min\left(1, \frac{Q_c^{(i)}(\mathcal{G}_{K_i})}{P_c^{(i)}(\mathcal{G}_{K_i})}\right).
\label{eq:pcg_acceptance}
\end{equation}
If $u_i \sim \text{Uniform}(0,1)$ and $u_i < r_i$, we accept the $group$ at position $i$.

\smallskip\noindent
\textbf{Residual sampling via thinning.}
In the reject branch, we must sample from the group-level residual
\begin{equation}
R_c^{(i)}(k) \propto \bigl[\, Q_c^{(i)}(\mathcal{G}_k) - P_c^{(i)}(\mathcal{G}_k) \,\bigr]_+,
\end{equation}
rather than from the token residual \(\bigl[\, q_i - p_i \,\bigr]_+\). We avoid enumerating all groups via thinning: sample \(y \sim q_i(\cdot)\); sample \(K \sim w_{k,y}\) over \(k \in S(y)\) (so \(K \sim Q_c^{(i)}\)); accept \(K\) with probability
\begin{equation}
a(K) = \Bigl[\, 1 - \frac{P_c^{(i)}(\mathcal{G}_K)}{Q_c^{(i)}(\mathcal{G}_K)} \,\Bigr]_+,
\end{equation}
otherwise resample \((y, K)\). Upon acceptance, emit a token within the accepted group according to the \(q\)-conditional
\begin{equation}
\Pr(z = t \mid K) = \frac{q_i(t)\, w_{K,t}}{Q_c^{(i)}(\mathcal{G}_K)}, \qquad t \in \mathcal{G}_K.
\end{equation}
This realizes the exact residual \(R_c^{(i)}\) at the group level while computing \(P_c^{(i)}(\mathcal{G}_K)\) and \(Q_c^{(i)}(\mathcal{G}_K)\) only for sampled groups via Eq.~\eqref{eq:PcQc}. The expected number of thinning trials is \(1/\mathrm{TV}\bigl(P_c^{(i)}, Q_c^{(i)}\bigr)\). In practice, we find that across all our experiments, the average number of trials does not exceed 3 trials.

\smallskip\noindent
\textbf{Guarantee and practical choice within group.}
\begin{proposition*}
Under the coupling \eqref{eq:group_sampling} with weights satisfying \eqref{eq:weights}, the speculative acceptance rule \eqref{eq:pcg_acceptance}, together with a reject branch that samples a group $K'$ from the residual distribution
\[
R_c^{(i)}(k)\ \propto\ \bigl[\,Q_c^{(i)}(\mathcal{G}_k)\;-\;P_c^{(i)}(\mathcal{G}_k)\,\bigr]_+,
\]
ensures that the emitted group at each step is distributed exactly according to $Q_c^{(i)}$.
\end{proposition*}

\begin{proof}
By construction, $\Pr(K_i=k)=P_c^{(i)}(\mathcal{G}_k)$. The acceptance rule \eqref{eq:pcg_acceptance} gives
\[
\Pr(\text{accept } K_i=k)
= \min\bigl(P_c^{(i)}(\mathcal{G}_k),\, Q_c^{(i)}(\mathcal{G}_k)\bigr).
\]
The reject branch emits group $k$ with probability proportional to $[\,Q_c^{(i)}(\mathcal{G}_k)-P_c^{(i)}(\mathcal{G}_k)\,]_+$. Summing the accepted and rejected contributions yields
\begin{multline}
\min\!\bigl(P_c^{(i)}(\mathcal{G}_k),\, Q_c^{(i)}(\mathcal{G}_k)\bigr)
+ [\,Q_c^{(i)}(\mathcal{G}_k)-P_c^{(i)}(\mathcal{G}_k)\,]_+ \\
= Q_c^{(i)}(\mathcal{G}_k).
\end{multline}
Thus the emitted group distribution equals $Q_c^{(i)}$ at each step.
\end{proof}

For efficiency, we retain the draft token $x_i$ as the group representative upon acceptance, preserving KV-cache validity. Crucially, this \emph{preserves the exact group-level distribution $Q_c$}, we accept the correct groups with the correct probabilities, only the intra-group selection follows $p$ rather than $q$. This trades token-level exactness for computational efficiency, which is appropriate for speech where group-level (acoustic) fidelity is paramount. Our adapted speculative decoding algorithm with overlap-aware PCG is in Algorithm \ref{alg:pcg}. 

\begin{algorithm}[t]
\small
\DontPrintSemicolon
\caption{Speculative Decoding with PCG}
\label{alg:pcg}
\KwIn{Draft $p$, target $q$, draft length $L_d$, groups $\mathcal{G}$, weights $w_{k,x}$}
\KwOut{$T_{\mathrm{in}}$}
$T_{\mathrm{in}}\leftarrow[\,]$\;
\For{$i=1$ \KwTo $L_d$}{
  Sample $x_i\sim p_i(\cdot)$\;
  Sample $\mathcal{G}_i\sim w_{k,x_i}$ over $k\in S(x_i)$\;
  Compute $P_c, Q_c$ for $\mathcal{G}_i$ from $p_i,q_i$ via Eq.~(\ref{eq:PcQc})\;
  $r_i\leftarrow \min\bigl(1,\,Q_c/P_c\bigr)$\;
  \lIf{$\mathrm{Uniform}(0,1)<r_i$}{append $x_i$ to $T_{\mathrm{in}}$}
  \lElse{
    $y \leftarrow \textsc{GroupResidualSampler}(p_i, q_i, w, \mathcal{G})$ \hfill{\footnotesize(samples from $R_c^{(i)} \propto [Q_c^{(i)}-P_c^{(i)}]_+$)}\;
    append $y$ to $T_{\mathrm{in}}$\;
    \textbf{break}\;
  }
}
\If{$|T_{\mathrm{in}}|=L_d$}{Sample $x_{L_d+1}\sim q_{L_d+1}(\cdot)$ and append to $T_{\mathrm{in}}$}
\Return{$T_{\mathrm{in}}$}
\end{algorithm}

\smallskip\noindent
\textbf{Complexity and memory.}
PCG adds only an $O(|G_i|)$ operations to compute $P_c^{(i)}$ and $Q_c^{(i)}$ for the sampled group; the $O(n)$ softmax over the vocabulary still dominates. 
In our setup, $\max_k |Average(G_k)|\approx 140$ ($|G_i|\ll n$).

ASGs are precomputed once and cached. Let $n = |\mathcal{V}|$. Storing the group membership lists is $O(n\bar{k})$ integers (plus optional reverse indices). For $n = 65{,}536$ and $\bar{k} \approx 140$ which correspond to the largest group size we experimented with, this is $\sim 9.2$M indices; with 32-bit indices, $\sim 37$MB ($\sim 19$MB with 16-bit indices). This is modest relative to model weights and KV-cache memory. Further reductions are possible via compression, sparsification (top-$K$ neighbors), or on-the-fly approximate search.

\begin{figure}[t]
    \centering
    \includegraphics[width=1\columnwidth]{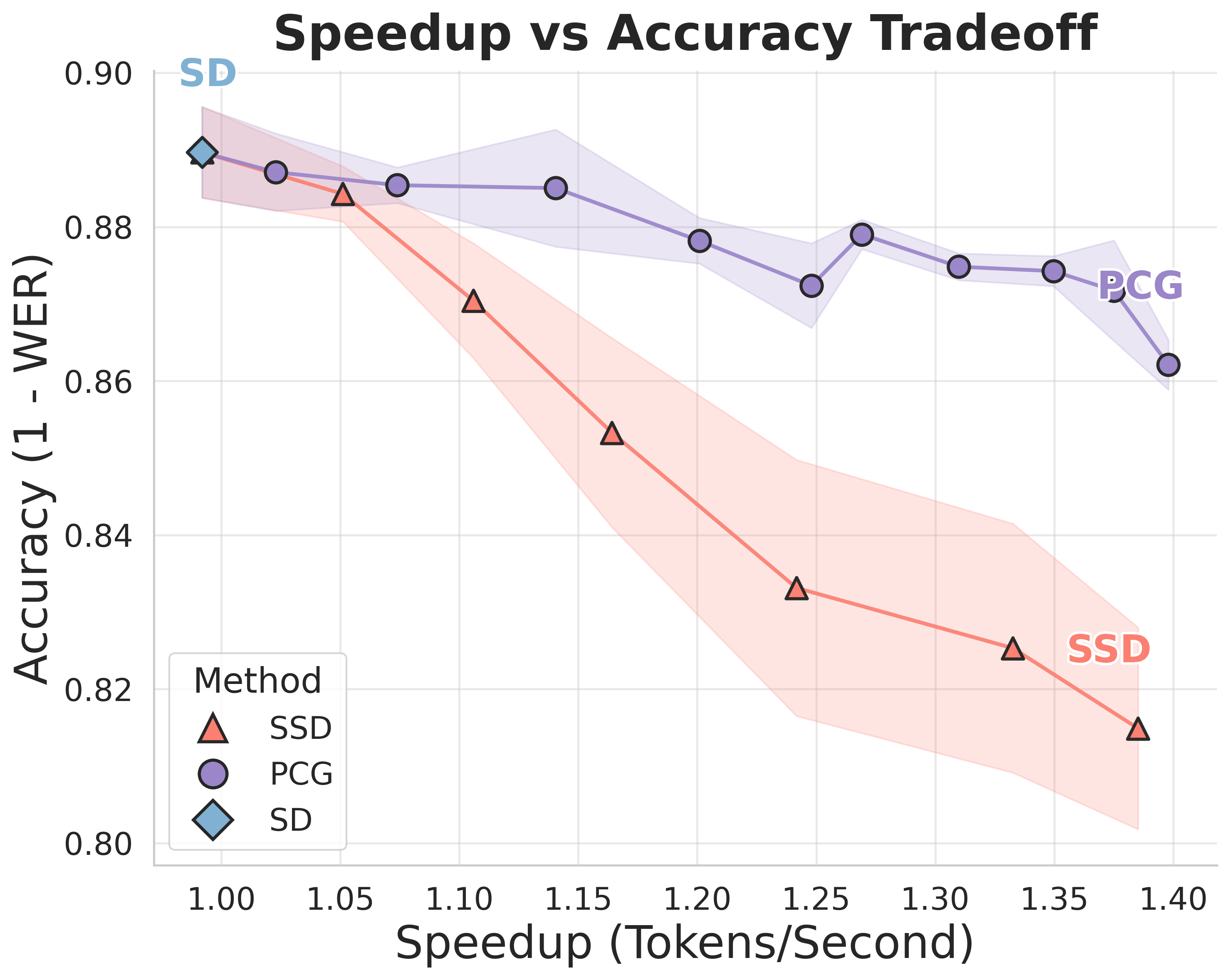}
    \caption{\textbf{Accuracy–Speedup Trade-off Curve on LibriTTS}. While the SSD baseline leads to a sharp drop in accuracy, PCG achieves a more favorable balance between speedup and accuracy. SD denotes standard speculative-decoding.}
    \label{fig:tradeoff_curve}
\end{figure}

\begin{table*}[t]
\centering
\renewcommand{\arraystretch}{1.05} %
\setlength{\tabcolsep}{7pt} %
\begin{tabular}{lcccccc}
\toprule
\textbf{Method}& \textbf{Speedup} $\uparrow$  & \textbf{WER} $\downarrow$ & \textbf{CER} $\downarrow$ & \textbf{Sim-O} $\uparrow$ & \textbf{NMOS} $\uparrow$ \\
\midrule
Draft & 5.2$\times$ & 52.8 $\pm$ 1.6& 41.4 $\pm$ 1.8 & 36.3 $\pm$ 1.1 & - \\
Target + SD & 0.98$\times$ & 11.1 $\pm$ 0.6 &  5.5 $\pm$ 0.5 & 43.7 $\pm$ 0.3 & 4.38 $\pm$ 0.88 \\
\midrule
Target + SSD~\cite{lin25h_interspeech} & 1.4$\times$ & 18.5 $\pm$ 1.9 &  11.6 $\pm$ 1.7 & 42.5 $\pm$ 0.4 & 3.78 $\pm$ 1.21 \\
Target + PCG & 1.4$\times$ & \textbf{13.8 $\pm$ 0.4} & \textbf{7.8$\pm$ 0.3} & \textbf{43.7 $\pm$ 0.1} & \textbf{4.09 $\pm$ 1.13} \\

\bottomrule
\end{tabular}
\caption{Evaluation results on LibriTTS (test-clean) for a fixed speedup factor. Metrics: Word Error Rate (WER), Character Error Rate (CER), Speaker similarity (Sim-O), Naturalness MOS (NMOS). NMOS scores on a 1--5 scale. }
\label{tab:libritts_results}
\end{table*}

\section{Experiments}
\subsection{Experimental Details}

\label{sec:exp_details}
We use LLaSA-8B~\cite{ye2025llasascalingtraintimeinferencetime} as the target model, built on LLaMA-8B~\cite{touvron2023llamaopenefficientfoundation}. It employs the X-codec2 tokenizer with FSQ-style discretization~\cite{mentzer2023finite} and a 65{,}536-entry codebook, and was chosen for its minimalist, strong architecture.

Following \cite{lin25h_interspeech}, the draft model is a 3-layer subset of the LLaSA target, initialized from LLaSA-8B parameters and trained on Libri-heavy~\cite{kang2024libriheavy50000hoursasr} (50{,}000 hours of read English). The draft was not heavily optimized (no knowledge distillation); all methods should improve with a stronger draft.

At inference, we use temperature 0.8, a speculation lookahead of 3 tokens, and precompute ASGs offline. Speedups are reported in tokens per second in wall-clock times. All experiments run on a single NVIDIA H100-80GB GPU.

\subsection{Evaluation Settings}
Following prior work \cite{du2024unicats,lin25h_interspeech}, we evaluate on 500 randomly sampled utterances from LibriTTS test-clean in a cross-sentence zero-shot speech-cloning setup. For automatic evaluation, we report Word Error Rate (WER) and Character Error Rate (CER) using a HuBERT-large ASR, and speaker-similarity scores computed with WavLM~\cite{chen2022wavlm}, averaged over 3 random seeds. Results report $\pm$ standard deviation.

For human perceptual assessment, we conduct a 5-point NMOS study (1 = unacceptable, 5 = excellent). Four participants each rate 85 randomly selected samples for each method. NMOS captures overall speech quality by jointly evaluating intelligibility and naturalness

We compare against two baselines: standard speculative decoding and SSD~\cite{lin25h_interspeech}, which relaxes acceptance by adding a constant bias to $r_i$ to trade accuracy for speed (analogous to our threshold $\theta$). To our knowledge, SSD is the only prior speech-specific method that avoids top-$k$ restrictions while attempting to approximate the target distribution; unlike our approach, it does not provide exact distributional guarantees, whereas our method preserves exact sampling at the ASG level. Finally, we include an ablation that replaces cosine similarity with Mel-spectrogram similarity to show that gains arise from similarity in the model’s token-embedding space, not merely from acoustic resemblance.

\begin{table}[t]
\centering
\begin{tabular}{@{}ccccc@{}}
\toprule
\textbf{Avg. $|\mathcal{G}|$} & \textbf{Swap \%} & \textbf{$\Delta$WER} & \textbf{$\Delta$Sim-O} & \textbf{Speedup$^\dagger$} \\ \midrule
1.0   & 0.0  & +0.000 & 0.000  & 1.0$\times$ \\ \midrule %
4.5   & 75.0 & +0.004 & -0.014 & 1.18$\times$ \\
10.3  & 91.4 & +0.007 & -0.027 & 1.23$\times$ \\
\bottomrule
\end{tabular}
\caption{Ablation on intra-group token substitution. Columns report average group size ($|\mathcal{G}|$), percentage of tokens swapped, changes in WER and Sim-O relative to ground truth, and relative decoding speedup. Minimal degradation supports the ASG assumption. $^\dagger$speedup, refers to the speedup measured when coupled with SD.}

\label{tab:ablation_seg}
\end{table}

\section{Results}
\label{sec:results}

Figure \ref{fig:tradeoff_curve} shows the speed–accuracy trade-off and Table \ref{tab:abl_similarity} shows that grouping by cosine similarity in the target embedding space yields lower WER than Mel-spectrogram similarity at the same speedup with comparable speaker similarity.
Table \ref{tab:libritts_results} summarizes the main results: Target+SD yields negligible speedup (0.98$\times$) while achieving the lowest WER/CER; SSD with 0.3 bias term, raises throughput to 1.4× but worsens WER/CER to 18.5/11.6. PCG, at the same 1.4× speedup, reduces WER/CER to 13.8/7.8 and improves speaker similarity and NMOS versus SSD. These gains are statistically significant (t-test $p<0.05$; PCG vs SSD $p=0.039$).

To test the interchangeability assumption within a ASG, we run a stress-test. Starting from target-generated token sequences, and for similarity thresholds $\theta$ used to construct ASGs, we replace every token that belongs to a multi-member group ($|G|>1$) with a uniformly sampled alternative from its ASG, ignoring the target’s token probabilities—so most positions are rewritten. We then measure degradation against original target utterances (Section \ref{sec:exp_details}).

Results (Table \ref{tab:ablation_seg}) support the hypothesis. With a relaxed threshold that would yield $\approx$1.23$\times$ speedup under SD+PCG, we swap 91.4\% of tokens (only 8.6\% unchanged), yet see only a small drop: WER +0.007 and speaker similarity -0.027 relative to the target baseline. This indicates that ASGs capture genuine acoustic similarity and justifies our group‑based acceptance criterion.

Finally, in Figure \ref{fig:lookahead_abl} we ablate the number of lookahead speculative tokens predicted by the draft model for SD with PCG. The plot shows that 3 tokens provide the best speedup.

\vspace{-0.3em}

\begin{table}[t]
\centering
\renewcommand{\arraystretch}{1.05}
\begin{tabular}{lccc}
\toprule
\textbf{Method} & \textbf{WER} $\uparrow$ &\textbf{CER} $\downarrow$ & \textbf{Sim-O} $\uparrow$ \\
\midrule
PCG (MEL)    &  13.9 & 7.9 &  \textbf{43.7} \\
PCG (Cosine) & \textbf{12.6} &  \textbf{6.6} &  \textbf{43.7} \\
\bottomrule
\end{tabular}
\caption{Comparison of PCG with cosine similarity and MEL-spectrogram similarity using WER and Sim-O metrics.}
\vspace{-0.2em}
\label{tab:abl_similarity}
\end{table}

\begin{figure}[t]
    \centering
    \includegraphics[width=1\columnwidth]{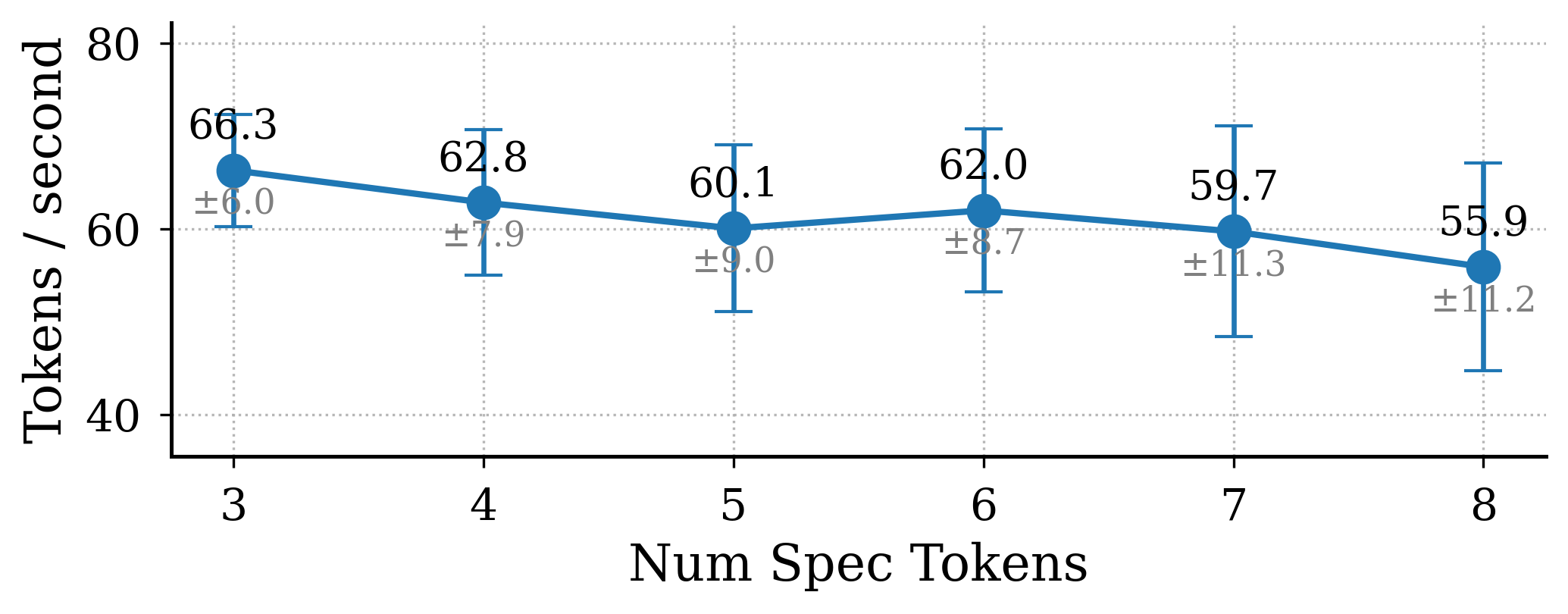}
    \caption{Ablation study on the effect of the number of speculated tokens on the speedup of SD with PCG.}
    \label{fig:lookahead_abl}
 \vspace{2.1em}  

\end{figure}

\section{Conclusion}

We introduced PCG to address inefficient speculative decoding in speech token generation. PCG replaces strict token-matching with a group-level acceptance criterion based on Acoustic Similarity Groups. Our method is theoretically grounded, guaranteeing the sequence of accepted acoustic concepts is sampled correctly from a coarse-grained version of the target distribution at each step. An efficient implementation of PCG achieves speedups while better preserving generation quality, outperforming baseline adaptations of speculative decoding to speech. This work demonstrates the effectiveness of principled, semantically-aware sampling for accelerating speech token generation.

\bibliographystyle{IEEEbib}
\bibliography{strings,refs}

\end{document}